\documentclass[12point]{amsart}
\usepackage{amsmath}
\usepackage{amscd}
\usepackage{breqn}
\usepackage{amssymb,amsmath,amscd,epsfig,amsthm,enumerate,import,color}

\newtheorem{theorem}{Theorem}[section]
\newtheorem{proposition}[theorem]{Proposition}
\newtheorem{lemma}[theorem]{Lemma}
\newtheorem{corollary}[theorem]{Corollary}
\theoremstyle{definition}
\newtheorem{definition}[theorem]{Definition}

\theoremstyle{remark}
\newtheorem{remark}[theorem]{Remark}

\numberwithin{equation}{section}
\newcommand{\be}{\begin{equation}}
\newcommand{\ee}{\end{equation}}

\newcommand{\bbC}{{\mathbb C}}

\newcommand{\bbR}{{\mathbb R}}
\newcommand{\bbT}{{\mathbb T}}

\newcommand{\calT}{{\mathcal T}}

\newcommand{\calV}{{\mathcal V}}

\newcommand{\calG}{{\mathcal G}}
\newcommand{\calU}{{\mathcal U}}

\newcommand{\calK}{{\mathcal K}}
\newcommand{\calL}{{\mathcal L}}

\newcommand{\calH}{{\mathcal H}}
\newcommand{\calE}{{\mathcal E}}

\newcommand{\calC}{{\mathcal C}}
\newcommand{\calR}{{\mathcal R}}
\newcommand{\calI}{{\mathcal I}}
\newcommand{\calS}{{\mathcal S}}

\newcommand{\calZ}{{\mathcal Z}}

\newcommand{\frakh}{{\mathfrak h}}

\DeclareMathOperator{\tr}{Tr}

\DeclareMathOperator{\Tr}{tr}

\newcommand{\inner}[2]{\langle#1,#2\rangle}

\newcommand{\norm}[1]{\lVert#1\rVert}

\newcommand{\wt}{\widetilde}

\newcommand{\h}{\hbar}

\newcommand{\hinv}{\hbar^{-1}}

\newcommand{\ave}{\text{ave}}

\DeclareSymbolFont{bbold}{U}{bbold}{m}{n}
\DeclareSymbolFontAlphabet{\mathbbold}{bbold}

\newcommand{\1}{\mathbbold{1}}

\begin{document}

\title[Asymptotics of eigenvalue clusters for the Landau Hamiltonian]{Perturbations of the Landau hamiltonian:  Asymptotics of eigenvalue clusters}

\author{G. Hernandez-Duenas}
\address{Instituto de Matem\'aticas, UNAM, Unidad Quer\'etaro}
\email{hernandez@im.unam.mx}
\thanks{G. Hernandez-Duenas  partially supported by project CONACYT Ciencia B\'asica A1-S-17634}
\author{S. P\'erez-Esteva}
\address{Instituto de Matem\'aticas, UNAM, Unidad Cuernavaca}
\email{spesteva@im.unam.mx}
\thanks{S. P\'erez-Esteva partially supported by the project PAPIIT-UNAM ININ106418}
\author{A. Uribe}
\address{Mathematics Department\\
University of Michigan\\Ann Arbor, Michigan 48109}
\email{uribe@umich.edu}
\thanks{A. Uribe supported by the NSF under Grant No. 1440140, while he was in residence at the Mathematical Sciences Research Institute in Berkeley, California, during the fall semester of 2019}
\author{C. Villegas-Blas}
\address{Instituto de Matem\'aticas, UNAM, Unidad Cuernavaca and ``Laboratorio Solomon Lefschetz" Mexico, Unidad Mixta Internacional del CNRS, Cuernavaca.}
\email{villegas@matcuer.unam.mx}
\thanks{C. Villegas-Blas partially supported by projects CONACYT Ciencia B\'asica 283531 and PAPIIT-UNAM  IN105718}

\date{}
\begin{abstract}
We consider the asymptotic behavior of the spectrum of the Landau Hamiltonian 
plus a rapidly decaying potential, as the magnetic field strength, $B$, tends to infinity.
After a suitable rescaling, this becomes a semiclassical problem where the
role of Planck's constant is played by $1/B$.  The spectrum of the
operator forms eigenvalue clusters.  We obtain a Szeg\H{o} limit theorem
for the eigenvalues in the clusters as a suitable cluster index and $B$ tend to 
infinity with a fixed ratio $\calE$.  The answer involves the averages of the potential
over circles of radius $\sqrt{\calE/2}$ (circular Radon transform). We also discuss
related inverse spectral results.
\end{abstract}
\maketitle 
\tableofcontents
\section{Introduction}

The Landau Hamiltonian, in the symmetric gauge, is the operator on $L^2(\bbR^2)$
\begin{equation}\label{}
\wt\calH_0 (B)= \frac{1}{2}\left(\frac{1}{i}\frac{\partial\ }{\partial x_1}+\frac{B}{2}\widehat{Q}_2
\right)^2 + \frac{1}{2}\left(\frac{1}{i}\frac{\partial\ }{\partial x_2} -\frac{B}{2}\widehat{Q}_1
\right)^2.
\end{equation}
It is the quantum Hamiltonian of a particle
on the plane subject to a constant magnetic field perpendicular to the plane and of intensity $B$.
Here $\widehat{Q}_j = $ multiplication by $x_j$ and we are taking the Planck's  parameter $\h=1$ at this point.  It is well known that the spectrum of the operator $\wt\calH_0 (B)$ is given by the set of Landau levels 
\begin{equation}
\lambda_q(B) = \frac{B}{2} \left(2q+1\right),  \;\; q=0,1,\ldots, 
\end{equation}
where each Landau level has infinite multiplicity. 

In \cite{PRV}, A. Pushnitski, G. Raikov and C. Villegas-Blas obtained a limiting eigenvalue distribution theorem for perturbations of the Landau Hamiltonian $\wt\calH_0 (B)$ given by a multiplicative  potential  $V:\bbR^2\rightarrow \bbR$.  More precisely, 
they studied perturbations of $\wt\calH_0$ of the form
\begin{equation}\label{}
\widetilde\calH(B) = \wt\calH_0(B) + V,
\end{equation}
where $V\in C(\bbR^2)$    
and $V$ is short-range, that is, it satisfies
\begin{equation}\label{srange}
\exists C>0,\rho>1\quad\text{such that}\quad \forall x\in\bbR^2
\quad |V(x)| \leq C \langle x\rangle^{-\rho}, \;\;\;\;   \langle x\rangle=\sqrt{1+|x|^2}.
\end{equation}
The authors  of \cite{PRV} show that, outside of a finite interval,  the spectrum of the operator  $\widetilde\calH(B)$ consists of  clusters of eigenvalues  around the Landau levels 
$\lambda_q(B)$.  The size of those  clusters   is $O(\lambda_q(B)^{-1/2})$ and, in the large energy limit  $q\rightarrow\infty$ with  $B$ fixed,  the scaled eigenvalues  of  
$\widetilde\calH(B)$ distribute inside those clusters according to a measure $d\mu$ 
which we now describe.  
Consider the function $\breve{V}:{\bbT}\times{\bbR}\rightarrow\bbR$, where $\bbT$ is the unit
circle, given by 
\begin{equation}
\breve{V}(\omega,b)= \frac{1}{2\pi} \int_{-\infty}^{\infty} V(b\omega + t\omega^{\perp})\;dt,   \;\;\; \omega=(\omega_1,\omega_2)\in\bbT, \;\;\;  \omega^\perp=(-\omega_2,\omega_1),       \;\;\; b\in\bbR   
\end{equation}
($\bbT\times\bbR$ parametrizes the manifold of straight lines on $\bbR^2$, and $\breve{V}(\omega,b)$ is the integral of  $V$ along the corresponding straight line.) 
Their main result is:
\begin{theorem}[Pushnitski, Raikov, Villegas-Blas]\label{Teo-PRV}
	Let $d\mu$ be the push-forward measure
	\[
	d\mu = \breve{V}_*\left( \frac{1}{2\pi}dm \right)
	\]
	 where $dm$ is the Lebesgue measure on ${\bbT}\times{\bbR}$.
Then, for $\rho\in C_0^{\infty}(\bbR\setminus\{0\})$ and $V$ as above, one has 
\begin{eqnarray}
\lim_{q\rightarrow\infty}\lambda_q(B)^{-1/2}  
\Tr \rho\left(\sqrt{\lambda_q(B)}\left(\widetilde\calH(B)-\lambda_q(B)\right) \right) = \int_\bbR \rho(\lambda) d\mu(\lambda).
\end{eqnarray}
\end{theorem}

In this paper we establish a different limiting eigenvalue distribution theorem for perturbations  of the Landau Hamiltonian, taking the semi-classical limit (or, equivalently, the large $B$ limit) as the classical energy of the unperturbed
Hamiltonian is fixed to a given value $\calE$.
Correspondingly, the result involves averages of $V$ along the classical orbits of the Landau problem  with energy $\mathcal E$.  This result is different than the one in Theorem 
\ref{Teo-PRV} because the latter is a result as  ${\mathcal E}\rightarrow \infty$ 
(although it can  be interestingly  rewritten as the ${\mathcal E}\rightarrow \infty$ limit of normalized integrals of V along circles with energy $\mathcal E$, see Eq. (1.16) in reference \cite{PRV}).

For our main result we will assume for simplicity that $V$ is Schwartz.
Introducing the small parameter 
\begin{equation}\label{}
\h = \frac 2B
\end{equation}
and factoring out $B^2/4$ in $\widetilde\calH_0$, one gets:
\begin{equation}\label{}
\calH(\hbar) := \h^{2}\widetilde\calH(B=2/\hbar)  = \calH_0(\hbar)+\h^2 V
\end{equation}
where 
\begin{equation}\label{}
\calH_0(\hbar) : = \h^{2}\widetilde\calH_0(B=2/\hbar) =\frac{1}{2}\left(\widehat{P}_1 +\widehat{Q}_2
\right)^2 + \frac{1}{2}\left(\widehat{P}_2 -\widehat{Q}_1
\right)^2 ,
\quad \widehat{P}_j=\frac{\h}{i}\,\frac{\partial\ }{\partial x_j}.  \label{defH_0}
\end{equation}
Therefore, up to an overall factor of $\h^2$, the large $B$ asymptotics of 
the operator $\widetilde\calH(B)$ is equivalent to 
the semi-classical asymptotics of the operator $\calH(\hbar)$, where $B$ and $\h$ are
related as above.

\medskip
The spectrum of $\calH_0(\hbar)$ consists of eigenvalues
\begin{equation}
E_k := \h(2k+1), \ k=0,1,\ldots ,  \label{energias} 
\end{equation}
each one infinitely degenerate.  (This will become transparent below.)  Since
``multiplication by $V$" is bounded in $L^2$, the spectrum of $\calH$
forms clusters of size $O(\h^2)$ around the $E_k$.  Each cluster contains
infinitely-many eigenvalues (counting multiplicity) which, by the stability of the essential spectrum,
only accumulate at the unperturbed eigenvalue $E_k$.
We will focus on the study of  the  distribution of eigenvalues inside clusters  of $\calH(\hbar=\frac{\calE}{2n+1})$ around a fixed classical energy $\calE$, in the semi-classical limit $\hbar\mapsto{0}$.  More precisely, let us take $\calE$ fixed and consider 
$\hbar$ taking discrete values along the sequence
\begin{equation}
\hbar=\frac{\calE}{2n+1},   \qquad n=0,1,\ldots.
\end{equation}
Consider the family 
${\mathcal F}_0 = \left\{ \calH_0(\hbar=\frac{\calE}{2n+1}) \;  | \;  n=0,1,\ldots \right\}$ of Schr\"odinger operators, and 
note that $\calE$ is an eigenvalue of  each member of the family ${\mathcal F}_0$ 
(corresponding to the quantum number $k=n$ in  (\ref{energias}) for each $n$).   

In this paper we study the distribution of eigenvalues of operators in
the family ${\mathcal F} = \left\{ \calH(\hbar=\frac{\calE}{2n+1}) \;  | \;  n=0,1,\ldots \right\}$ that
cluster around $\calE$ when $n\mapsto\infty$ or, equivalently,  $\hbar\mapsto{0}$.

To state our main theorem,  
consider the classical Hamiltonian $H_0: T^*\bbR^2 \mapsto \bbR$
of a charged particle moving on the  plane $\left\{(x_1,x_2,0  \;  | \; x_1,x_2 \in \bbR)\right\}$ under the influence of the constant magnetic field $(0,0,2)$ corresponding to the quantum Hamiltonian $\calH_0(\hbar)$:  
\begin{equation} 
H_0({\bf x},{\bf p}) = \frac{1}{2}\left(p_1+x_2\right)^2 +  \frac{1}{2}\left(p_2-x_1\right)^2,\quad 
{\bf x}=(x_1,x_2), \;\;\; {\bf p}=(p_1,p_2).
\end{equation}
It can be shown that,  for a fixed value $\calE$ of the energy $H_0$, the classical orbits of $H_0$ in configuration space are circles with radius $\sqrt{\frac{\calE}{2}}$ and period $\pi$. Any given point in 
$\bbR^2$ can be the center of one of those circles. More explicitly, if we denote by $t$ the time evolution parameter,  we have 
\begin{eqnarray}\label{eqciirculo}
x_1(t) & =&     \frac{P_2}{\sqrt{2}} + \sqrt{\frac{\calE}{2}}\sin(2(t+\phi))\nonumber \\
x_2(t)  &=&     \frac{X_2}{\sqrt{2}} + \sqrt{\frac{\calE}{2}}\cos(2(t+\phi)) 
\end{eqnarray}
where $P_2:=\left(x_1+p_2\right)/\sqrt{2}$ and $X_2:=\left(x_2-p_1\right)/\sqrt{2}$ are integrals of motion whose particular values are determined by the initial conditions ${\bf x}(0)\;$,  
$\frac{d{\bf x}}{dt}(0)$, and   the equations 
$\left(p_1(0),p_2(0)\right)= \left(\frac{dx_1}{dt}(0)-x_2(0),\frac{dx_2}{dt}(0)+x_1(0)\right)$.  The angle 
$\phi$ is a solution of  the equation $\exp(2\imath\phi)= \frac{1}{ \sqrt{2 \calE}}\left(\;p_1(0)+x_2(0) -\imath (p_2(0)-x_1(0))\;\right)$. 
 
\medskip
Our main result is the following:

\begin{theorem}\label{Main}
Assume that $V\in\calS(\bbR^2)$ (the Schwartz class), 
and fix a positive number $\calE>0$.  Let $\h\to 0$ along the sequence such that
\begin{equation}\label{regimen}
\h(2n+1) = \calE,\qquad n=0,1,\ldots.
\end{equation}  Then for any $\rho\in C_0^\infty(\bbR)$ such that $\rho(\lambda)/\lambda$ is continuous we have
\begin{equation}
\lim_{\hbar\to{0}} \hbar \Tr\left( \rho\left(\frac{\calH(\hbar)-\calE}{\hbar^2}\right)\right) = \frac{1}{2\pi}\int_{\bbR ^2}\rho\left(\widetilde{V}(X_2,P_2;\calE)\right)dX_2\;dP_2,
\end{equation}
\end{theorem}

\noindent where $\widetilde{V}(X_2,P_2;\calE)$ denotes the average of $V$ along the circle with center $(\frac{P_2}{\sqrt{2}},\frac{X_2}{\sqrt{2}})$ and radius $ \sqrt{\frac{\calE}{2}}$
given by (\ref{eqciirculo}), namely
\begin{equation}
\widetilde{V}(X_2,P_2;\calE) = \frac{1}{\pi} \int_0^\pi V\left( \frac{P_2}{\sqrt{2}} + \sqrt{\frac{\calE}{2}}
\sin(2t)\; , \;  \frac{X_2}{\sqrt{2}} + \sqrt{\frac{\calE}{2}}\cos(2t) \right) dt.
\end{equation}


Analogous results in the literature include the asymptotics of eigenvalue clusters for the Laplacian plus
a potential on spheres and other Zoll manifolds (see \cite{W} for the seminal work on this type of theorems), bounded perturbations of the n dimensional isotropic harmonic oscillator, $n\geq{2}$ \cite{OV},    
and for both bounded and unbounded perturbations of the quantum hydrogen atom Hamiltonian, \cite{UV}, \cite{HV}, \cite{AHV1}. In all  those cases, the result involves averages of the perturbation along the classical orbits of the unperturbed Hamiltonian with a fixed energy. 
As we have mentioned, in some sense the result in Theorem \ref{Teo-PRV}  corresponds to the
case $\calE = \infty$, in which case the average of $V$ should be taken over circles of 
infinite radius, that is, straight lines. 

The paper is organized as follows.  We begin in the next section by showing that one
can replace the perturbation by an ``averaged" version of it.  For this we re-examine 
estimates derived in \S 4 of \cite{PRV}, in order to keep track of the dependence on $B$. 
Using this result, in \S 3 we reduce the problem to studying the spectrum of a one-dimensional
semi-classical
pseudo-differential operator.  A complication is that the Weyl symbol of this operator
is given by matrix elements of another operator which depends on parameters.  This requires
an analysis of the reduced operator which is the subject of \S 4.  We complete the proof
of Theorem \ref{Main}
in \S 5, and in \S 6 we obtain some inverse spectral results, assuming that we know
the spectrum of $\widetilde{\calH}(B)$ for all $B$.  In the appendices we review some technical
results that are needed in the analysis of the reduced operator.

\section{The main lemma}

In this section we will show that, to leading order, the moments of the spectral measures 
of the eigenvalue clusters can be
computed by ``averaging" the perturbation; see Lemma \ref{MainLemma} below.

We follow closely the arguments in \cite{PRV}, \S 4.
For $q=0,1,\ldots$, let us denote by $\wt{P}_q(B)$ the orthogonal projector  with range the eigenspace of the operator $\wt\calH_0 (B)$ with eigenvalue $\lambda_q(B)=\frac{B}{2} \left(2q+1\right)$.  
We begin by the following result which is actually lemma 4.1 in \cite{PRV}, but with the dependence on the intensity $B$ of the magnetic field made explicit:  
 
 \begin{lemma}\label{estnormaVR_0V}
 Assume  the potential $V$ satisfies condition (\ref{srange}).   
Given $q=0,1,2,\ldots$, consider the positively oriented  circle $\Gamma_q$  with center $\lambda_q(B)$ and radius $\frac{B}{2}$.   Then for  all $z\in\Gamma_q$ and any integer $\ell>1$, $\ell>1/(\rho-1)$, we have
\begin{equation}
\sup_{z\in\Gamma_q} \; \| \;  |V|^{1/2} \wt{R}_0 (z;B) |V|^{1/2} \; \| _{\mathcal{L}_\ell} = B^{\frac{-\ell+1}{2\ell}}O(q^{\frac{-\ell+1}{2\ell}}\log q)
\end{equation}
where $ \wt{R}_0 (z;B)$ denotes the resolvent operator $(\wt\calH_0 (B)-z)^{-1}$and $\mathcal{L}_\ell$ denotes the Schatten ideal on $L^2(\bbR^2)$.
\end{lemma}
\begin{proof}
 First we write
 \begin{equation}
  \wt{R}_0 (z;B) =  \sum_{k=0}^{\infty} \frac{\wt{P}_k(B)}{\lambda_k(B)-z}.
\end{equation}
From part (ii) of theorem 1.6 in reference \cite{PRV}, we know that for $\ell > 1((\rho-1)$ and $B_0>0$ there exist $C=C(B_0,\ell)$ such that 
\begin{equation}
\sup_{k\geq{0}}  \; \sup_{B\geq{B}_0} \; \lambda_k(B)^{(\ell-1)/(2\ell)} \; B^{-1} \; \|  \wt{P}_k(B) V \wt{P}_k(B)\|_\ell 
\leq C \sup_{x\in\bbR^2} \left(1+|x|^2\right)^{\rho/2} |V(x)|.
\end{equation}
Thus we have $\| \;  |V|^{1/2} \wt{P}_k(B)   |V|^{1/2} \; \| _{\ell} = C B \lambda_k(B)^{-(\ell-1)/(2\ell)}$ where, from now on, we  are denoting different constants whose values are not relevant for our purposes by the same letter $C$. Defining 
$\nu\equiv(\ell-1)/(2\ell)$ we obtain:
 \begin{eqnarray}
&&\| \;  |V|^{1/2} \wt{R}_0 (z;B) |V|^{1/2} \; \| _{\ell}  \leq  C B \sum_{k=0}^{\infty} 
\frac{\lambda_k(B)^{-\nu}}{|\lambda_k(B)-z|} \nonumber \\
&&\leq C  B^{-\nu}\sum_{k=0}^{\infty} 
\frac{(k+1)^{-\nu}}{| \; (2k+1)-[(2q+1)+\exp(\imath{\theta})] \; |} \nonumber \\
&&\leq C  B^{-\nu} \left[ \sum_{k=0}^{q-1} \frac{(k+1)^{-\nu}}{a-2(k+1)} + 
(q+1)^{-\nu} + \sum_{k=q+1}^{\infty} \frac{(k+1)^{-\nu}}{2(k+1) - c} 
\right] \label{inefirstsum},
\end{eqnarray}
where $a=2q+1$, $c=2q+3$,  $z=\frac{B}{2} \left(2q+1\right) + \frac{B}{2}\exp(\imath{\theta})$, with $\theta\in[0,2\pi]$.
  Let $f(x)=\frac{(x+1)^{-\nu}}{a-2(x+1)}$, $x\in[0,q-1]$. Note that  $f(x)$ has a minimum at  $x_0=(\nu(a-2)-2)/(2(\nu+1))$.   Since   $1/4\leq\nu<1/2$ then we have 
 \begin{eqnarray}
&&\sum_{k=0}^{q-1} \frac{(k+1)^{-\nu}}{a-2(k+1)} \leq f(0) + f(q-1) + \int_{0}^{q-1}f(x)dx \nonumber \\
 && \phantom{xxxxxx} = O(q^{-1}) + O(q^{-\nu}) +   \int_{0}^{q-1}f(x)dx = O(q^{-\nu}) +   \int_{0}^{q-1}f(x)dx. 
\end{eqnarray}
The integral in the last equation can be estimated as follows:  
\begin{eqnarray}
 \int_{0}^{q-1}f(x)dx \leq \frac{1}{a-2(x_0+1)}\int_{0}^{x_0}(x+1)^{-\nu}dx + (x_0+1)^{-\nu}\int_{x_0}^{q-1}\frac{1}{a-2(x+1)}dx \nonumber \\
 = O(q^{-1})O(q^{-\nu+1}) + O(q^{-\nu})O(\log(q)) = O(q^{-\nu})O(\log(q)),
 \end{eqnarray}
where we have used $x_0 = O(q)$.  Thus the first sum  in (\ref{inefirstsum}) is $O(q^{-\nu})O(\log(q))$.

Now let 
$g(x)=\frac{(x+1)^{-\nu}}{2(x+1) - c}$,  $x\in[q+1,\infty)$.  Since $g$ is a decreasing function, then we have
\begin{eqnarray}
&&\sum_{k=q+1}^{\infty} \frac{(k+1)^{-\nu}}{2(k+1) - c}  \leq g(q+1) + \int_{q+1}^{\infty} g(x)dx  = O(q^{-\nu}) + \int_{q+2}^{\infty} \frac{y^{-\nu}}{2y-c}dy \nonumber \\
&&=  O(q^{-\nu}) + \frac{1}{2}\left( \frac{c}{2}\right)^{-\nu} \int_{\frac{2q+4}{2q+3}}^{\infty}\frac{w^{-\nu}}{w-1}dw 
\nonumber \\ && =  O(q^{-\nu}) + \frac{1}{2}\left( \frac{c}{2}\right)^{-\nu}
\left[ \int_{\frac{2q+4}{2q+3}}^{2}\frac{w^{-\nu}}{w-1}dw 
+  \int_{2}^{\infty}\frac{w^{-\nu}}{w-1}dw \right] \nonumber \\
&&   \leq  O(q^{-\nu}) + \frac{1}{2}\left( \frac{c}{2}\right)^{-\nu}\left[ \int_{\frac{2q+4}{2q+3}}^{2}\frac{1}{w-1}dw +\tilde{C}]\right]   \nonumber \\ && =     O(q^{-\nu}) + \frac{1}{2}\left( \frac{c}{2}\right)^{-\nu} O (\log(q)) = O(q^{-\nu}\log(q)),
\end{eqnarray}
where $\tilde{C}$ is the constant $ \int_{2}^{\infty}\frac{w^{-\nu}}{w-1}dw $.    This concludes the proof. 

\end{proof}

Now we are ready to establish a crucial ``averaging lemma", which will allow us to compute asymptotically
the moments of the eigenvalue clusters of the operator $\calH(\h)$.
For $n=0,1,\ldots$,  denote by $P_n(\hbar)$ the orthogonal projector with range  the eigenspace of the operator  $\calH_0(\hbar)$ 
with eigenvalue $E_n=\hbar(2n+1)$.   
We have the following: 
\begin{lemma}\label{MainLemma}
Fix $\calE>0$, and let $\1_{(\calE -\h\,,\,\calE+\h)}$ denote the characteristic function of the interval
	$(\calE-h\,,\,\calE+\h)$ .  Then for each $\ell>1$,  $\ell > 1/(\rho-1)$ we have 
\begin{eqnarray}\label{mainLemma}
&&\tr\left[ \left(\calH(\hbar)-\calE I\right)^\ell \1_{(\calE -\h\,,\,\calE+\h)}(\calH(\hbar)) \right]  = \nonumber  \\ &&    \;\;\;\;\;\; \;\;\;\;\;\;   \;\;\;\;\;\; 
\h^{2\ell} \tr\left[ \left( \; P_n(\hbar) \; V \;  P_n(\hbar)\;\right)^\ell\right] + o(\hbar^{2\ell-1})
\end{eqnarray}
as $\h\to 0$ and $n\to\infty$ in such a way that $\h(2n+1) = \calE$.  
\end{lemma}

\begin{proof}
The proof follows  the corresponding proof of lemma 1.5 in reference \cite{PRV}, but using the estimate provided by lemma  \ref{estnormaVR_0V}. 
Throughout the proof we will assume the following identities:
\begin{equation}\label{identidades}
B=\frac{2}{\hbar}\quad\text{and}\quad \hbar={\calE}/({2n+1}).
\end{equation}

Let us denote by $R(\eta;\hbar)=\left(\calH(\hbar)- \eta I \right)^{-1}$ the resolvent operator associated with  the operator 
$\calH(\hbar)$ at the point $\eta\in{\mathbb C}$, whenever it is well defined. 
%
%
If ${\mathcal C}_{\calE}$ denotes the positively oriented circle with center ${\calE}$ and radius $\hbar$, we can write:
\begin{equation}
\left(\calH(\hbar)-\calE I\right)^\ell \1_{(\calE -\h\,,\,\calE+\h)}(\calH(\hbar))  = \frac{-1}{2\pi\imath}\int_{{\mathcal C}_{\calE}} \left(\eta-\calE \right)^\ell R(\eta;\hbar) \; d\eta. \label{expansion1}
\end{equation}
Keeping in mind (\ref{identidades}),  notice that 
\begin{equation}
R(\eta;\hbar) = \hbar^{-2}  \wt{R}\left(z;B\right), \quad \wt{R}(z;B):=(\wt\calH (B) - z I )^{-1}  \label{resolvents-prop}
\end{equation}
provided
\[
 z=\lambda_n\left(B\right) + \frac{B}{2}\exp(\imath\theta).
\]
%
%
Therefore, equation (\ref{expansion1}) can be written as
\begin{equation}
\left(\calH(\hbar)-\calE I\right)^\ell \1_{(\calE -\h\,,\,\calE+\h)}(\calH(\hbar))  = \frac{-\hbar^{2\ell}}{2\pi\imath}\int_{ \Gamma_n} 
\left(z - \lambda_n \left(B\right)\right)^\ell \wt{R}\left(z,B\right) \; dz,  \label{expansion2}
\end{equation}
where $\Gamma_n$ denotes the positively oriented circle with center $ \lambda_n \left(B\right)$
 and radius $B/2$.
The integral on the right-hand side of equation (\ref{expansion2}) has been studied in section 4 of reference \cite{PRV}  where, in particular, it is shown that  such an integral is trace class and the following expansion holds:
\begin{eqnarray}
&&{\tr}\left[\left(\calH(\hbar)-\calE I\right)^\ell \1_{(\calE -\h\,,\,\calE+\h)}(\calH(\hbar))\right] =   \nonumber 
 \hbar^{2\ell} \tr\left[ \left( \; P_n(\hbar) \; V \;  P_n(\hbar)\;\right)^\ell \right]  +	\nonumber \\
&&+ \frac{\ell\hbar^{2\ell}}{2\pi\imath}\sum_{j=\ell+1}^{\infty}\frac{(-1)^j}{j}\int_{\Gamma}\left(z - \lambda_n \left(B\right)\right)^{\ell-1} \tr \left[ \left(V \wt{R}_0\left(z,B\right)\right)^j\right] \; dz, \nonumber \\ \label{inftracesum}
\end{eqnarray}
where $\wt{R}_0\left(z,B\right)=\left(\wt\calH_0 (B)  - z I \right)^{-1}$ 
and we have used that $\wt{P}_n(B)=P_n(\hbar)$ are actually the same operator, always assuming (\ref{identidades}). 


As in reference \cite{PRV}, let us write $V=|V|^{1/2} {\rm sign}(V)  |V|^{1/2}$.  Then we have for $j\ge\ell$: 
\begin{eqnarray}
&&\left |
\tr\left[\left(V \wt{R}_0(z,B)]\right)^j\right] \right | = \left |  \tr\left[ \left(  {\rm sign}(V)    |V|^{1/2} \wt{R}_0(z,B)  |V|^{1/2}  \right)^j   \right]     \right |\nonumber  \\ 
 && \phantom{xxx}\leq \left\|  \left(  {\rm sign}(V)    |V|^{1/2} \wt{R}_0(z,B)  |V|^{1/2} \right)^j \right \|_{\calL_1}
   \leq \left\|    {\rm sign}(V)    |V|^{1/2} \wt{R}_0(z,B)  |V|^{1/2}  \right \|_{\calL_j}^j \nonumber \\ 
&& \phantom{xxx}  \leq \left\|  |V|^{1/2} \wt{R}_0(z,B)  |V|^{1/2}  \right \|_{\calL_j}^j  \leq  \left\|  |V|^{1/2} \wt{R}_0(z,B)  |V|^{1/2}  \right \|_{\calL_\ell}^j, \nonumber
\end{eqnarray}
where we have used H\"older's inequality with $\frac{1}{j}+\frac{1}{j}+\ldots\ + \frac{1}{j}=1$.  
We summarize: for $j\geq\ell$,
\begin{equation}\label{traceestj}
\left |
\tr\left[\left(V \wt{R}_0(z,B)]\right)^j\right] \right | \leq \left\|  |V|^{1/2} \wt{R}_0(z,B)  |V|^{1/2}  \right \|_{\calL_\ell}^j.
\end{equation}
Using thus inequality and lemma  \ref{estnormaVR_0V}  with $B=\frac{2}{\hbar}$ and $q=n$  we can show that  the second term on  the right-hand side of equation (\ref{inftracesum})  can be estimated by
\begin{eqnarray}
C \hbar^{2\ell}\sum_{j=\ell+1}^{\infty} \hbar^{1-\ell}\left(\hbar^{\frac{\ell-1}{\ell}}\log(\hbar^{-1})\right)^j\hbar^{-1} 
\leq &&\nonumber \\ 
C \hbar^\ell\left(\hbar^{\frac{\ell-1}{\ell}}\log(\hbar^{-1})\right)^{\ell+1}
 \sum_{j=0}^{\infty} \left(\hbar^{\frac{\ell-1}{\ell}} \log(\hbar^{-1})\right)^j \leq &&
C \hbar^{2\ell-1}\left(\hbar^{\frac{\ell-1}{\ell(\ell+1)}}\right)^{\ell+1} = o(\hbar^{2\ell-1}), \nonumber 
\end{eqnarray}
since $\ell>1$ and $\sum_{j=0}^{\infty} \left(\hbar^{\frac{\ell-1}{\ell}} \log(\hbar^{-1})\right)^j $ is bounded for $\hbar$ sufficiently small.


\end{proof}
 
\section{Reduction to a one-dimensional pseudo-differential operator}

As we will see in this section, 
the analysis of the asymptotics of the eigenvalue clusters
in the regime that we are interested in amounts to analyzing the spectrum of
an $\h$-pseudo-differential operator on the real line.

\subsection{A preliminary rotation}  We begin by conjugating the unperturbed
operator $\calH_0$ by a suitable
unitary operator that separates variables and converts $\calH_0$ into a
one-dimensional harmonic oscillator tensored with the identity operator on $L^2(\bbR)$.

\begin{proposition}
	Let $\calU: L^2(\bbR^2)\to L^2(\bbR^2)$ be a metaplectic operator quantizing
	the linear canonical transformation $\calT: T^*\bbR^2\to T^*\bbR^2$ 
	such that if $X_j = x_j\circ \calT,$ $P_j = p_j\circ \calT$, $j=1,2$ then
	\begin{equation}\label{laTransf}
	\left\{
	\begin{array}{cc}
	P_1 = \frac{1}{\sqrt{2}} (p_1+x_2), & X_1 =\frac{1}{\sqrt{2}} ( x_1-p_2)\\ \\
	P_2 = \frac{1}{\sqrt{2}} (p_2+x_1), & X_2 =\frac{1}{\sqrt{2}} (x_2-p_1).
	\end{array}
	\right.
	\end{equation}
	Then
	\begin{equation}\label{}
	\calU^{-1}\circ \calH_0 \circ \calU = 
	-\h^2\,\frac{\partial^2\ }{\partial x_1^2} +  x_1^2 =: \calH_1.
	\end{equation}
\end{proposition}
\begin{proof}
	It is known that for metaplectic operators the Egorov theorem is {\em exact}:  
	For any symbol $a:T^*\bbR^2\to\bbC$, if $\text{Op}^W(a)$ denotes Weyl
	quantization of $a$,
	\begin{equation}\label{exactEgorov}
	\calU^{-1}\circ \text{Op}^W(a)\circ \calU = \text{Op}^W( a\circ \calT).
	\end{equation}
	Therefore the full symbl of $\calU^{-1}\circ \calH_0 \circ \calU$ 
	is just $P_1^2+X_1^2$.
\end{proof}

It is now clear that the spectrum of $\calH_0$, which is to say,
the spectrum of $\calH_1$, consists of the eigenvalues
$\h(2n+1)$, $n=0,\, 1,\ldots$ with infinite multiplicity.  
Let us denote by 
\[
\{e_n(x_1), n=0,\ldots\}
\]
an orthonormal eigenbasis of the one-dimensional quantum 
harmonic oscillator 
\begin{equation}\label{1dOsc}
\calZ := -\h^2\,\frac{\partial^2\ }{\partial x_1^2}+  x_1^2.
\end{equation}
Then the $n$-th eigenspace of $\calH_1$ is the infinite-dimensional space
\begin{equation}\label{}
\calL_n = \{ e_n(x_1)f(x_2)\;;\; f\in L^2(\bbR)\}.
\end{equation}

\bigskip
Let us now take $V:\bbR^2\to \bbR$ say in the symbol class $S^{-\rho}$:
\[
\forall \alpha\ \exists C_\alpha\ \forall x\in\bbR^2\qquad  |\partial^\alpha_x V(x)|\leq 
C_\alpha (1+|x|)^{-\rho-|\alpha|}.
\]
We will denote by
\begin{equation}\label{}
K:= \calU^{-1}\circ \text{Op}^W(V)\circ \calU = \text{Op}^W( V\circ \calT)
\end{equation}
the conjugate by $\calU$ of the operator of multiplication by $V$.  On the right-hand side
we are abusing the notation and denoting again by $V$ the pull-back of $V$ to $T^*\bbR^2$.
Partially inverting (\ref{laTransf}), one has
\begin{equation}\label{}
x_1 = \frac{1}{\sqrt{2}} (X_1+P_2),\quad x_2=\frac{1}{\sqrt{2}} (X_2+P_1)
\end{equation}
and therefore the function $W:= V\circ \calT$ is
\begin{equation}\label{}
W(X,P):=\left( V \circ \calT\right) (X,P) = V\left(\frac{1}{\sqrt{2}} (X_1+P_2), \frac{1}{\sqrt{2}} (X_2+P_1)\right).
\end{equation}
The Schwartz kernel of the operator $K$ is 
\begin{dmath}\label{}
\calK(X,Y) = \frac{1}{(2\pi\h)^2}\iint e^{i\hinv(X-Y)\cdot P}\,  V\left(\frac{1}{2\sqrt{2}}
(X_1+Y_1+2P_2, X_2+Y_2+2P_1)\right)\, dP_1\,dP_2.
\end{dmath}
and
\begin{equation}\label{}
\calU^{-1}\circ \left(\calH_0+\h^2 V\right)\circ \calU  = \calH_1+\h^2 K.
\end{equation}
This is the operator we will analyze.

\subsection{Averaging}
For ease of notation we will re-name the $(X,P)$ variables back to $(x,p)$.

Let us consider the unitary $2\pi$-periodic one-parameter group of operators
\begin{equation}\label{}
\calV(t):= e^{-it\calH_0/\h}
\end{equation}
For each $t$ this is a metaplectic operator associated with the
graph of the linear canonical transformation
\begin{equation}\label{}
\phi_t: T^*\bbR^2\to T^*\bbR^2\qquad \phi_t(x_1,p_1\;;\;x_2,p_2) = 
(\frakh_t(x_1,p_1)\;;\; x_2,p_2)
\end{equation}
where $\frakh_t: T^*\bbR \to T^*\bbR$ is the one-dimensional harmonic
oscillator of period $\pi$ (the Hamilton flow of $x_1^2+p_1^2)$.

\medskip
Let us define
\begin{equation}\label{}
K^\ave := \frac{1}{\pi}\int_0^{\pi} \calV(-t) K\, \calV(t)\, dt.
\end{equation}
For each $n=1,2,\ldots$ denote by $\calL_n$ the eigenspace of $\calH_1$ of eigenvalue 
$E_n =\h(2n+1)$, and let
\[
\Pi_n: L^2(\bbR^2)\to \calL_n
\]
be the orthogonal projector. 
Then it is not hard to verify that $\left[ K^\ave, \Pi_n\right] = 0$ and
that
\begin{equation}\label{sandwiches}
\Pi_n\,K^{\ave}\, \Pi_n =\Pi_n\, K\, \Pi_n.
\end{equation}
Therefore
\begin{equation}\label{}
\forall \ell=1, 2, \cdots\quad
\left(\Pi_n K \Pi_n\right)^\ell  = \Pi_n \left[K^\ave\right]^\ell \Pi_n.
\end{equation}

\begin{lemma}
	$K^\ave$ is a pseudo-differential operator of order zero.  In fact
	\begin{equation}\label{}
	K^\ave = \text{Op}^W(W^\ave)
	\end{equation}
	where $W^\ave$ is the function
	\begin{equation}\label{}
	W^\ave(x,p) = \frac{1}{\pi}\int_0^{\pi} W(\phi_t(x,p))\, dt.
	\end{equation}
\end{lemma}
\begin{proof}
	This is once again due to the fact that $\calV(t)$ is a metaplectic operator
	for each $t$, and for such operators Egorov's theorem is exact.
\end{proof}

\medskip
For future reference we compute $W^\ave$ in terms of $V$ when $x_1=0$.
(This determines $W^\ave$, by $\phi_t$ invariance.)     A trajectory
of the flow $\phi_t$ is
\[
x_1(t) = \sin(2t)p_1(0),\quad p_1(t) = \cos(2t) p_1(0).
\]
The energy of the trajectory is $E= p_1(0)^2$.  Then
\[
W^\ave(x_1=0, x_2, p_1(0), p_2) = \frac 1\pi
\int_0^\pi V\left(u(t) , v(t)\right)\, dt,
\]
where
\[
u(t) = \frac{1}{\sqrt{2}}\left(\sin(2t)p_1(0) + p_2\right), 
\quad v(t) = \frac{1}{\sqrt{2}}\left(x_2 + \cos(2t)p_1(0)\right)
\]
is a parametrization of the circle
\begin{equation}\label{theCircle}
\calC_{x_2, p_2, E} :=
\left\{ (u,v)\;;\;
\left( u-\frac{p_2}{\sqrt{2}}\right)^2 + \left(v-\frac{x_2}{\sqrt{2}}\right)^2 = \frac 12 E
\right\}.
\end{equation}
We see that it is then natural to regard $W^\ave$ as a circular Radon transform of $V$.
More precisely, let us define
\begin{equation}\label{radonDef}
\forall \xi\in\bbR^2, \ E>0	\qquad \widetilde{V}(\xi; E) := 
\frac{1}{2\pi}\int_{S^1} V\left(\check{\xi} + \sqrt{E/2}\,\omega\right) ds(\omega)
\end{equation}
where $s$ is arc length and $\check{\xi} = \frac{1}{\sqrt{2}}(p_2,x_2)$ if $\xi = (x_2, p_2)$.
Then
\begin{equation}\label{wave}
W^\ave(x,p) = \widetilde{V}(x_2, p_2; x_1^2+p_1^2).
\end{equation}

\medskip
We now fix $\calE>0$, and let $\h$ tend to zero along the sequence such that
\begin{equation}\label{bohrS}
E_n= \h\,(2n+1) = \calE,\qquad n=1,2,\ldots.
\end{equation}
By Lemma \ref{MainLemma}, the moments of the shifted eigenvalue clusters around $\calE$
of $\calH_1+\h^2K$ are, to leading order, the same as the moments of the eigenvalues of the operator
\[
\Pi_n\,K^{\ave}|_{\calL_n} : \calL_n \to \calL_n.
\] 

\begin{lemma}
	For each $n=1,2,\ldots$
	there is an operator $T_n: L^2(\bbR)\to L^2(\bbR)$ such that
	\begin{equation}\label{}
	\forall f\in L^2(\bbR)\qquad \Pi_n\,K^{\ave}(e_n\otimes f) = e_n\otimes T_n(f).
	\end{equation}
\end{lemma}
It is clear that
\begin{equation}\label{opT}
T_n(f)(x_2) = \int \overline{e_n}(x_1)\,K^{\ave}(e_n\otimes f)(x_1,x_2)\, dx_1.
\end{equation}
Note that we also have that
$\Pi_n\,K(e_n\otimes f) = e_n\otimes T_n(f)$,
by (\ref{sandwiches}).

\begin{definition}
	We call the sequence of operators $(T_n)$ the {\em reduction} of $K$ at level $\calE$.
\end{definition}
We emphasize that the interest of the operator $T_n$ is that, by the previous
considerations and by Lemma \ref{MainLemma},
\begin{equation}\label{}
\tr\left[ \left(\calH(\hbar)-\calE I\right)^\ell \1_{(\calE -\h\,,\,\calE+\h)}(\calH(\hbar)) \right]  =  
\h^{2\ell} \tr (T_n^\ell) + o(\h^{2\ell-1})
\end{equation}
as $\h\to 0$ along the values (\ref{bohrS}).

\section{Analysis of the reduced operator}
Our  goal in this section 
is to show that, for our purposes, $T_n$ can be replaced by
a semi-classical pseudo-differential operator whose symbol is  
$\widetilde{V}(x_2, p_2;\calE)$.  

From now on the parameters $\h$ and $n$ are assumed to be related by the condition
(\ref{bohrS}).

We will also assume that $V$ is a Schwartz function.

\subsection{The Weyl symbol of $T_n$}

Since $K^{\ave}$ is the Weyl quantization of the function (\ref{wave}),
namely
\[
W^\ave(x,p) = \widetilde{V}(x_2, p_2; x_1^2+p_1^2), 
\]
one has
\begin{align*}
K^{\ave}(e_n\otimes f)(x_1,x_2) = & \\
=\frac{1}{(2\pi\h)^2}\int  e^{i\hinv (x-y)p}&\,
\widetilde{V}\left(\frac{x_2+y_2}{2}, p_2; \left(\frac{x_1+y_1}{2}\right)^2+p_1^2\right)\,
e_n(y_1)\,f(y_2)\, dy\, dp.
\end{align*}
Therefore, after changing the order of integration, we can rewrite (\ref{opT}) as
\begin{equation}\label{}
T_n(f)(x_2) = \frac{1}{(2\pi\h)}\int e^{i\hinv (x_2-y_2)p_2}\, \Phi\left(\frac{x_2+y_2}{2}, p_2,n\right)
\, f(y_2)\,dy_2\, dp_2,
\end{equation}
where 
\begin{dmath}\label{}
\Phi(x_2,p_2,n)= \frac{1}{(2\pi\h)}\int e^{i\hinv (x_1-y_1)p_1}\, 
\widetilde{V}\left(x_2, p_2; \left(\frac{x_1+y_1}{2}\right)^2+p_1^2\right)\,
\overline{e_n}(x_1)\,e_n(y_1)\, dy_1\,dp_1\,dx_1.
\end{dmath}
From this we immediately obtain:
\begin{lemma}
Let, for each $\xi := (x_2,p_2)$, ${B}_{\xi}$ be the operator
which is the Weyl quantization of the ($\h$-independent) function 
\begin{equation}\label{}
{b}_{\xi}(x_1,p_1):=\widetilde{V}\left(\xi; x_1^2+p_1^2\right).
\end{equation}
Then the Weyl symbol of $T_n$ is
\begin{equation}\label{}
\Phi(\xi, n)= \inner{{B}_{\xi}(e_n)}{e_n}.
\end{equation}
\end{lemma}

%

\begin{remark}
	The function
	$b_{\xi}(x_1,p_1)$ is  Schwartz
	as a function of the variables $(x_1,p_1)$, with estimates
	uniform as $\xi$ ranges on compact sets.
\end{remark}

\medskip
As a function of $(x_1, p_1)$ the function $b_{\xi}(x_1,p_1)$ is radial, that is, it is
a function of $x_1^2+p_1^2$.  We will make use of the following result on the Weyl
quantization of a radial function on the plane, whose proof we present in Appendix A,
following \cite{DHS}.
\begin{proposition}
	Let $a\in C^\infty(\bbR^2)$ be a (Schwartz) radial function, that is
	\[
	a(x,p) = \rho(r), \quad r=\sqrt{x^2+p^2},
	\]
	and let $A:=a^W(x,\h D)$ be its Weyl quantization. 
	Then $\forall n$ $e_n$ is an eigenfunction of $A$ with eigenvalue
	\begin{equation}\label{autoValor}
	\lambda_n = \frac{(-1)^n}{\h} \int_0^\infty \rho(\sqrt{u}) e^{-u/\h}\,L_n(2u/\h)\, du,
	\end{equation}
	where $L_n$ is the normalized $n$-th Laguerre polynomial.
\end{proposition}
From this we get the following explicit expression for the Weyl symbol of $T_n$:
\begin{equation}\label{explicitLaguerre}
\Phi(\xi, n)=\frac{(-1)^n}{\h}\int_0^\infty \widetilde{V}(\xi,u)e^{-u/\h}L_n(2u/\h)du.
\end{equation}

\begin{proposition}  
	For each $\h$ (and therefore $n$) the function $\Phi$ is Schwartz if $V$ is.
\end{proposition}
\begin{proof}
	In view of (\ref{explicitLaguerre}), since
	 $n$ and $\h$ are fixed, it suffices to prove that the function
\[
f(\xi)=\int_0^\infty \widetilde{V}(\xi,u)e^{-u/\h}u^mdu
\]
	is Schwartz for any positive power $m$.
	Split the integral defining $f$ in the form 
	$f(\xi)=\int_0^{\vert \xi\vert/2} \widetilde{V}(\xi,u)e^{-u/\h}u^mdu+\int_{\vert\xi\vert/2}^{\infty} \widetilde{V}(\xi,u)e^{-u/\h}u^mdu.$
	Since $V$ is Schwartz, then $\vert V(y)\vert \lesssim \langle y\rangle^{-M}$ for any $M$. 
Therefore, by the definition of the Radon transform (\ref{radonDef}),
	\begin{equation*}
	\left| \int_0^{\vert \xi\vert/2}\widetilde{V}(\xi,u)e^{-u/\h}u^mdu\right|\lesssim\langle \xi\rangle^{-M}.
	\end{equation*}
	On the other hand
	\begin{equation*}
	\left|\int_{\vert\xi\vert/2}^{\infty}\widetilde{V}(\xi,u)e^{-u/\h}u^m du\right|du\leq\Vert V\Vert_\infty \int_{\vert\xi\vert/2}^{\infty}e^{-u/\h}u^m du=O(\langle\xi\rangle^{-\infty}).
	\end{equation*}
	Since $\partial_\xi \widetilde{V}(\xi,u)= \widetilde{\partial_\xi V}(\xi,u)$, 
	we can repeat the argument on all derivatives of $\widetilde{V}$ and conclude that $\Phi(\cdot, n)\in\mathcal{S}$.
\end{proof}

\subsection{Localization}

In this section we  cut $\Phi$ (and therefore $T$) in two pieces, 
and show that one can neglect one of the pieces.
Let $M>\calE$ and $\chi\in C_0^\infty (\bbR)$ such that $\chi\equiv 1$ on $[0,M]$ and $\chi(t)\equiv 0$ for $t>2M$, and
for each $\xi\in\bbR^2$ let
\begin{equation}\label{}
f_{\xi}(t) := \widetilde{V}(\xi; t)\,\chi(t).
\end{equation}
Let us now define
\begin{equation}\label{Phione}
\Phi_1(\xi,n) = \inner{{F}_{\xi}(e_n)}{e_n}
\end{equation}
where $F_{\xi}$ is the Weyl quantization of the function 
\begin{equation}\label{theFunction}
(x_1,p_1)\mapsto  f_{\xi}(x_1^2+p_1^2),
\end{equation}
and let
\begin{equation}\label{Phitwo}
\Phi_2(\xi,n) =\Phi(\xi,n) -\Phi_1(\xi,n).
\end{equation}
We denote by $T^{(i)}_n$ be the Weyl quantization of $\Phi_i(\cdot, n)$, $i=1, 2$.
These functions are Schwartz for each $n$ (by the same proof that $\Phi$ is Schwartz), and  
$T_n= T^{(1)}_n+T^{(2)}_n$.

\medskip
Next we show that $T^{(2)}_n$ is negligible.
\begin{theorem}\label{T2chiquito}
	Let $V\in\mathcal{S}$. Then there exists $M>\mathcal{E}$ such that if the support of the cut-off $\chi$ above satisfies $\text{supp}(\chi)\subset[0, M]$,
	then $\Vert T^{(2)}_n\Vert_{\mathcal{L}^1} =O(\h^\infty )$ provided $\h(2n+1)=\mathcal{E}$.
\end{theorem}
\begin{proof}
	Using (\ref{explicitLaguerre}), 
\begin{equation*}
\Phi_2(\xi,n)=\frac{(-1)^n}{\h}\int_0^\infty \widetilde{V}(\xi,u)(1-\chi(u))e^{-u/\h}L_n(2u/\h)du.
\end{equation*}
We want to apply the known trace-norm estimate
\begin{equation*}
\Vert \text{Op}^W (\Phi_2)\Vert_{\mathcal{L} ^1}\leq \frac{C}{\h} \max_{\vert \beta\vert\leq 3}\int_{\mathbb{R} ^2} \vert  \partial_\xi^\beta \Phi_2(\xi)\vert d\xi
\end{equation*}
(see \cite{CR} chapter 2, Theorem 5 ).  	

First notice that, from the definition of the Radon transform (\ref{radonDef}), 
\[
\partial_\xi^\beta \widetilde{V}(\xi,u) = 2^{-|\beta|/2}\widetilde{(\partial^{\beta^T} V)}(\xi, u)
\]
where $\beta^T = (\beta_2, \beta_1)$ if $\beta = (\beta_1, \beta_2)$.
Therefore
\[
\int_{\bbR^2} \left| \partial_\xi^\beta \widetilde{V}(\xi,u) \right|\, d\xi  = 
 2^{-|\beta|/2}\,\Vert \partial^{\beta^T} V\Vert_{L^1 }
\]
for all $u>0$.  Since $\text{supp}(\chi)\subset[0, M]$, it follows that
\begin{equation}\label{}
	\int_{\mathbb{R} ^2} \vert  \partial_\xi^\beta \Phi_2(\xi)\vert d\xi\leq 
\frac{1}{2^{|\beta|/2}\h}\Vert \partial^{\beta^T} V\Vert_{L^1 }
\int_M^\infty e^{-u/\h}\vert L_n(2u/\h)\vert du.
\end{equation}

	Next we will use the representation of the Laguerre polynomials as a residue, namely
	\[
	L_n(t)=\frac{1}{2\pi i}  \oint_{\vert z\vert=r} \frac{e^{-t\frac{z}{1-z}}}{(1-z)z^n}dz,
	\]
	which holds for $0<r<1$.
	Since $\Re \frac{z}{1-z}=\frac{r\cos(\theta)-r^2}{\vert 1-z\vert^2}$ where $z=re^{i\theta}$,
	\begin{align*}
	\int_M^\infty e^{-u/\h}\vert L_n(2u/\h)\vert du &\leq \frac{r}{2\pi (1-r)r^n}\int_M^\infty e^{-u/\h}\int_0^{2\pi}e^{-\frac{2u}{\h}\frac{r\cos (\theta)-r^2}{(1-r)^2}}d\theta\,du\\ 
	&\leq \frac{r}{(1-r)r^n}\int_M^\infty e^{-u/2\h}du,
	\end{align*}
	if $r$ is small enough. 
	
	Now, since $\h(2n+1)=\mathcal{E}$
	\begin{align*}
	e^{-u/2\h}/r^n=r^{1/2}e^{-u/2\h+\log (1/r)\frac{\mathcal{E}}{2\h}}\lesssim e^{-u/4\h}
	\end{align*} 
	provided $u\geq M=-2\log (r)   \mathcal{E}$.
	Thus for this choice of $M$,
	\begin{equation*}
	\int_M^\infty e^{-u/\h}\vert L_n(2u/\h)\vert du=O(\h^\infty)
	\end{equation*}
	and the the proof is complete.
\end{proof}

\subsection{Estimates on $\Phi_1$}

This section is devoted to the proof of the following


\begin{theorem}\label{simboloReducido}
	As $\h\to 0$ along the sequence (\ref{bohrS}), 
\begin{equation}\label{phiUnoEst}
	\Phi_1(\xi,n) = \widetilde{V}(\xi, \calE) + \h^2 \calR(\xi,\h)
\end{equation}
where $\calR$ is a Schwartz function of $\xi$ which is $O_\calS(1)$, meaning that
\begin{equation}\label{}
\forall \alpha,\, \beta,\  \exists C,\,\h_0>0 \quad \text{such that}\quad
\forall \h\in(0,\h_0)\quad \sup_{\xi\in\bbR^2} \left|\xi^\alpha \partial_\xi^\beta \calR(\xi,\h)  \right| \leq C.
\end{equation}
\end{theorem}
\begin{proof}
	Recall that $\Phi_1(\xi,n)$ is defined by (\ref{Phione}), where the operator $F_\xi$ is the Weyl
	quantization of the radial function $f_\xi(x_1^2+p_1^2)$.  Consider the first-order
	Taylor expansion of $f_\xi(t)$ at $t=\calE$
	\[
	f_\xi(t) = \widetilde{V}(\xi, \calE) + (t-\calE)f'_\xi(\calE) +(t-\calE)^2R_\xi(t)
	\]
	where $R_\xi$ is Schwartz, as follows from the explicit formula (see Appendix \ref{Taylor})
	\begin{equation}\label{}
	R_\xi(t) = \int_0^1\int_0^1 uf_\xi''(uv(t-\calE) + \calE)\, du\, dv.
	\end{equation}  
	Denote the Weyl quantization of a function $a$ as $a^W$, and let
	\[
	\calI = x_1^2 + p_1^2.
	\]
Since $\inner{(\calI-\calE)^W(e_n)}{e_n} = 0$,
	(\ref{phiUnoEst}) holds with 
	\begin{equation}\label{needToEstimate}
	\calR(\xi,\h) = \frac{1}{\h^2}\inner{\left[(\calI-\calE)^2 R_\xi(\calI)\right]^W(e_n)}{e_n}.
	\end{equation}
	Consider now the triple Moyal product with remainder
	\[
	(\calI-\calE)\# (\calI-\calE)\# R_\xi(\calI) = (\calI-\calE)^2 R_\xi(\calI) +  \h^2 S_\xi(\calI, \h).
	\]
	Since $\inner{\left[(\calI-\calE)\# (\calI-\calE)\# R_\xi(\calI)\right]^W(e_n)}{e_n} = 0$, we obtain
	that (\ref{needToEstimate}) equals
	\[
		\calR(\xi,\h) = -\inner{S_\xi(\calI, \h)^W(e_n)}{e_n}.
	\]
{\sc Claim:}  Every partial derivative $\partial_{(x_1,p_1)}^{\alpha}$ of  
$S_\xi(\calI, \h)$ is $O(\langle \xi\rangle^{-N})$ for any $N$, uniformly in $\xi$ and $\h\leq\h_0$.
	  
	  To see this we use the  following fact  (see in \cite{M} Theorem 2.7.4 and its proof): 
	  If $a\in S(m)$ and $b\in S(m')$, then the Moyal product $a\#b$ is in $S(m+m')$ and its asymptotic expansion  is uniform in $ S(m+m')$  (here $f\in S(m) $ if and only if $\Vert\langle\xi\rangle^{-m}\partial^\alpha (\xi)\Vert\leq C_\alpha$ for every $\alpha$). 
	  More precisely, if   
$$a\#b \sim \sum_j \h^j c_j,$$ then for every $j$  
$$\left\langle x\right\rangle ^{-(m+m') }\left|\partial^{\alpha}c_j(x)\right|\leq C_{j,\alpha}$$
$C_{j,\alpha}$ depending only on $$\sup_{x\in\mathbb{R}^n, |\beta|\leq M}\left\langle x\right\rangle ^{-m }\left|\partial^{\beta}a(x)\right|\quad \text{and} \sup_{x\in\mathbb{R}^n, |\beta|\leq M}\left\langle x\right\rangle ^{-m' }\left|\partial^{\beta}b(x)\right|,$$
where $M= M(\alpha,j)$.  As a consequence of the stationary phase method, the same is true for each 
remainder of the asymptotic expansion of $a\#b$.
The claim follows by 
applying this argument  to combinations of $\calI-\calE$ and  $ R_\xi(\calI)$ and using that   for every $\alpha$ 
	$$\vert\partial_{(x_1,p_1)}^\alpha R_\xi(\calI)\vert\leq \frac{C_{\alpha,N}}{\langle \xi\rangle^N}.$$

	Next we use the estimate
	 (\cite[Ch. 2, Th. 4]{CR}
	
\begin{equation}\label{CaldVill}
	\Vert a^W\Vert\lesssim \sup_{\vert\alpha\vert\leq 5}\Vert \partial^{\alpha} a\Vert_\infty
	\end{equation}	
to conclude that 
	$$\vert\calR(\xi,\h) \vert\leq C_N\langle \xi\rangle^{-N}.$$
	Finally,  to estimate the derivatives $\partial^\alpha \calR(\xi,\h)$
	 we simply notice that $\partial^\alpha \calR(\xi,\h) $ replaces $\calR(\xi,\h) $ when we study the Landau problem with the potential $\partial^\alpha V$. With the same calculations  we conclude that
	
	$$\vert\partial^\alpha \calR(\xi,\h) \vert\leq C_{N,\alpha}\langle \xi\rangle^{-N},$$
	that is, $\calR(\cdot,\h)=O(1)$ in $\mathcal{S}(\mathbb{R}^2)$ for $\h\leq\h_0.$

\end{proof}

\begin{remark} Using again that  (see \cite[Ch. 2, Th.5]{CR}) 

\begin{equation*}
\Vert  a^W\Vert_{\calL_1}\lesssim\frac{1}{\h}\sum_{\vert\gamma\vert\leq 5}\Vert \partial^\gamma a\Vert_{L^1(\bbR ^2)},
\end{equation*}
we have by Theorem \ref{simboloReducido} and \eqref{CaldVill}  that 

\begin{equation}\label{Rnormas}
\Vert  \calR(\cdot,\h)^W\Vert_{\calL_1}\leq\frac{C}{\h}\quad\text{and}\quad \Vert  \calR(\cdot,\h)^W\Vert\leq C
\end{equation} 
for $\h\leq\h_0,$ and also
\begin{equation*}
\Vert  T_n^{(1)}\Vert_{\calL_1}\leq\frac{C}{\h}\quad\text{and}\quad \Vert T_n^{(1)}\Vert\leq C.
\end{equation*}
 It follows that 
\begin{equation}\label{tnnormas}
\Vert  T_n\Vert_{\calL_1}\leq\frac{C}{\h},\quad \Vert T_n\Vert \leq C.
\end{equation} 
\end{remark}
The previous Theorem and the symbol calculus imply:
\begin{corollary}
For any $\ell = 1,2,\ldots$, as $n\to\infty$ and with $\h(2n+1)=\calE$
\begin{equation}\label{potencias}
\tr\left( T^{(1)}_n \right)^\ell= \frac{1}{2\pi\h}\int_{\bbR^2} \widetilde{V}(\xi,\calE)^\ell\, d\xi + O(1).
\end{equation}	
\end{corollary}

\begin{proof} 
$\left( T^{(1)}_n \right)^\ell= \left(\widetilde{V}(\cdot, \calE)^W + \h^2 \calR(\cdot,\h)^W\right)^\ell,$ hence
\begin{equation}
\tr \left( T^{(1)}_n \right)^\ell=\tr \left[\left(\widetilde{V}(\cdot, \calE)^W\right)^\ell\right] +\h^2 \tr G_{\h},
\end{equation}
where $G_{\h}$ is a finite sum of terms each consisting of the product of a non negative power of $\h$ and  an operator of the form $S_1S_2\cdots S_m$, with $S_i\in\lbrace \widetilde{V}(\cdot, \calE)^W, \calR(\cdot,\h)^W\rbrace$ 
Using that 
\begin{equation}\label{productotrace}
	\Vert AB\Vert_{\calL_1}\leq \Vert A\Vert\Vert B\Vert_{\calL_1}
	\end{equation}
we conclude using \eqref{Rnormas} that 

\begin{equation}
\Vert\h^2 \tr G_{\h}\Vert_{\calL_1}=O(\h).
\end{equation}
Therefore, by the symbol calculus 
\begin{align*}
\tr \left( T^{(1)}_n \right)^\ell&=\tr \left[\left(\widetilde{V}(\cdot, \calE)^W\right)^\ell \right]+O(\h)\\
                                              &= \frac{1}{2\pi\h}\int_{\bbR^2} \widetilde{V}(\xi,\calE)^\ell\, d\xi + O(1).
\end{align*}

\end{proof}
\section{Proof of Theorem \ref{Main}}

We finally complete the proof of Theorem \ref{Main}, that is:
\begin{theorem}
	Suppose that $V\in\calS$. Let $g\in C[-M,M]$ and $f(t)=tg(t)$.
	Then 
	\begin{equation*}
	\lim_{\h\to 0}2\pi\h\tr f(T_n)= \int_{\bbR^2}f\left(\widetilde{V}(x,p;\calE)\right) dx\,dp.
	\end{equation*}
\end{theorem}
\begin{proof}
	We start with $f(t)=t^{m}$.  When $m=1$ the result follows from Theorem \ref{T2chiquito} and 
	(\ref{potencias}).  
	If $m\geq 2$, we have that $(T_n)^{m}=(T_n^{(1)})^{m}+S$ where the operator $S$ is a finite sum of operators of the form $S_1S_2\cdots S_m$, with $S_i\in\lbrace T_n^{(1)},T_n^{(2)}\rbrace$ and where at least one factor $S_{i_0}$
	is equal to $T_n^{(2)}$.
	Hence from  Theorem \ref{T2chiquito},
	\begin{equation*}
	\vert \tr( S_1S_2\cdots S_m)\vert = \left| \tr\left( \prod_{i_0\leq i\leq m}S_i\prod_{1\leq j< i_0}S_j\right)\right|
	\end{equation*}  
	\begin{equation*}
	\leq C\Vert T_n^{(2)}\Vert/\h^{l}=O(\h^{\infty})
	\end{equation*} 
	 for  some power $l$,
	where we have used several times \eqref{productotrace}.
	  
	We conclude that
	\begin{equation}
	\h\tr\, (T_n)^m=\h\tr\,(T_n^{(1)})^m+O(\h^{\infty}),
	\end{equation}
	and from (\ref{potencias}), 
	\begin{equation*}
	2\pi\h\tr\, (T_n)^{m}=\int_{\bbR^2}\widetilde{V}(x,p;E)^{m} dx\,dp + O(\h).
	\end{equation*} 
	Therefore the theorem holds when  $f$ is a polynomial.
	
To prove the general case, let $(p_k)$ be a sequence of polynomials converging uniformly to $g$
	on $[-M,M]$ as $k\to\infty$.
	Notice that
	\[
		\h\vert\tr f(T_n)-\tr p_k(T_n)T_n\vert\leq \h\, \Vert g(T_n)-p_k(T_n)\Vert\,\Vert T_n\Vert_{\calL_1}
	\leq C\Vert g-p_k\Vert_\infty,
	\]
	where the last inequality uses \eqref{tnnormas}.
	Hence 
	\begin{equation} \label{aprox1}
\tr\left(f(T_n)\right) =\lim_{k\rightarrow\infty} tr \left(T_np_k(T_n)\right)
	\end{equation}
	uniformly for $0\leq \h\leq\h_0$.
	
	Likewise,
	\begin{equation}\label{aprox2}
\tr\left(f(T_n^{(1)})\right) =\lim_{k\rightarrow\infty} \tr \left(T_n^{(1)}p_k(T_n^{(1)})\right).
	\end{equation}
	On the other hand, letting $q_k(t)=tp_k(t)$, we have
	
	\begin{equation*}
	\vert q_k(\widetilde{V}(x,p;E)))\vert\leq C\widetilde{V}(x,p;E),
	\end{equation*}
	and so by the dominated convergence theorem
	\begin{equation}\label{integrales}
	\int_{\bbR^2}f(\widetilde{V}((x,p);E))dx\,dp=\lim_{k\rightarrow\infty}\int_{\bbR^2}q_k(\widetilde{V}((x,p);E))dx\,dp.
	\end{equation}
	Finally we can estimate:
	\begin{align*}
	\left|\h\tr f(T_n)-\int_{\bbR^2}f\left(\widetilde{V}(x,p;E)\right) dx\,dp\right|&\leq\vert \h\tr(f(T_n)-\h\tr(q_k(T_n))\vert\\
	&+\vert\h\tr(q_k(T_n))-\h\tr(q_k(T_n^{(1)}))\vert\\
	&+\left|\h\tr(q_k(T_n^{(1)}))-\int_{\bbR^2}q_k(\widetilde{V}(x,p;E))dx\,dp\right| \\
	&+\left|\int_{\bbR^2}q_k(\widetilde{V}(x,p;E))dx\,dp-\int_{\bbR^2}f(\widetilde{V}(x,p;E))dx\,dp\right|\\
	&=: I_1+I_2+I_3+I_4.
	\end{align*}
	By \eqref{aprox1},\eqref{aprox2} and \eqref{integrales}, given $\epsilon>0$ we can find $k$ large enough so that $I_1+I_2+I_4<\epsilon$ for all $\h\in (0, h_0]$, and since $I_3$ tends to $0$ as $\h\rightarrow 0$
	(since the $q_k$ are polynomials), the proof is complete. 
\end{proof}

\section{An inverse spectral result}

Let us assume we know the spectrum of $\widetilde{\calH}_0(B)+V$ 
with $V\in\calS(\bbR^2)$,  for all $B$ in a neighborhood of infinity. 
 What can we say about $V$?  In this section we prove:

\begin{theorem}\label{resultadoInverso}
	If $V$ and $V'$ are two isospectral potentials (in the sense above) in the Schwartz class, 
	then $\forall s\in\bbR$ their Sobolev $s$-norms are equal:
	\[
	\norm{V}_s = \norm{V'}_s.
	\]
\end{theorem}

We will proceed as in \cite{GU} and use that the spectral data above determine the spectral invariant
\begin{equation}\label{}
I(r) = (2\pi)^2\,\int_{\bbR^2} \calR_{r}(V)^2(y)\, dy
\end{equation}
for all $r>0$, where  $\calR_{r}(V)(y)$ is the Radon transform of $V$, namely  the integral transform that averages
$V$ over the circle of radius $r$ and center $y$ (hence in the notation of section 3, $\tilde{V}(y,\calE)=\calR_{\sqrt{\calE /2}}(V)(\check{y}), $ $\check{y} = \frac{1}{\sqrt{2}}(p,x)$ if $y = (x, p))$.

\begin{lemma}
	Let $J_0$ denote the zeroth Bessel function.  Then
	\begin{equation}\label{}
	\calR_{r}(V)(y) = \frac{1}{(2\pi)^2}\,\int e^{iy\cdot\xi}\, J_0(r|\xi|)\, \widehat{V}(\xi)\, d\xi
	\end{equation}
	where $\widehat{V}$ is the Fourier transform of $V$.
\end{lemma}
\begin{proof}
	By the Fourier inversion formula it suffices to compute
	\[
	\calR_{r} (e^{-ix\cdot\xi})(y) = \frac{1}{2\pi }\int_0^{2\pi} \exp\left(
	-i[\xi_1(y_1 + r\cos(\theta)) + \xi_2(y_2+r\sin(\theta))]\right)\, d\theta =
	\]
	\[
	= \frac{e^{-iy\cdot \xi}}{2\pi }\int_0^{2\pi} \exp\left(
	-ir[\xi_1 \cos(\theta) + \xi_2\sin(\theta)]\right)\, d\theta .
	\]
	Let us now introduce polar coordinates for $\xi$,
	\[
	\xi = |\xi|\, u, \quad u=\langle \sin(\phi)\,,\, \cos(\phi)\rangle.
	\]
	Then
	\[
	\xi_1 \cos(\theta) + \xi_2\sin(\theta) = |\xi|\, \sin(\theta+\phi),
	\]
	and therefore
	\begin{equation}\label{}
	\calR_{r} (e^{-ix\cdot\xi})(y)= \frac{e^{-iy\cdot \xi}}{2\pi }\int_0^{2\pi}
	e^{-ir|\xi|\,\sin(\theta+\phi)}\, d\theta = 
	\frac{e^{-iy\cdot \xi}}{2\pi }\int_0^{2\pi}
	e^{-ir|\xi|\,\sin(\theta)}\, d\theta.
	\end{equation}
	However it is known that
	\[
	J_0(s) = \frac{1}{2\pi}\int_0^{2\pi} e^{-is\sin(\theta)}\, d\theta, 
	\]
	so we obtain
	\begin{equation}\label{}
	\calR_{r}(e^{-ix\cdot\xi})(y)= e^{-iy\cdot \xi}\, J_0(r|\xi|).
	\end{equation}
\end{proof}

Using Parseval's theorem, we immediately obtain:
\begin{corollary}
	\begin{equation}\label{laEcu}
	I(r) = \int_{\bbR^2} J_0(r|\xi|)^2\, \left|\widehat{V}(\xi)\right|^2\, d\xi.
	\end{equation}
\end{corollary}

\medskip
Let us now introduce polar coordinates $(\rho,\phi)$ on the $\xi$ plane,
and let us define
\begin{equation}\label{}
W(\rho) := \rho^2\,\int_0^{2\pi} \left|\widehat{V}(\rho^{-1}\cos(\phi),\, \rho^{-1}\sin(\phi))\right|^2\, d\phi
\end{equation}
and  
\[
K(s) = J_0(s)^2.
\]
Then (\ref{laEcu}) reads
\begin{equation}\label{laEcuReads}
I(r) = \int_0^\infty K(r\rho)\, W(\rho^{-1})\, \frac{d\rho}{\rho}.
\end{equation}
In other words, $I(r)$ is the convolution of $K$ and $W$ in the multiplicative
group $(\bbR^{+}, \times)$.

\begin{corollary}
	For each $\rho>0$, the integral
	\begin{equation}\label{normas2}
	\int_0^{2\pi} \left|\widehat{V}(\rho\cos(\phi),\, \rho\sin(\phi))\right|^2\, d\phi
	\end{equation}
	of $|\widehat{V}|^2$ over the circle centered at the origin and of radius $\rho$
	is a spectral invariant of $V$.
\end{corollary}
\begin{proof}
	By (\ref{laEcuReads}), the Mellin transform of $I$ is the product of the Mellin
	transforms of $K$ and $W$.  Since $K$ and its Mellin transform are analytic,
	and the Mellin transform of $W$ is continuous, 
	this determines the Mellin transform of $W$, and hence determines $W$.
\end{proof}

Theorem 6.1 follows from this, as

\[
\norm{V}_s^2 =\int_{\bbR^2}(1+\vert \xi\vert^2)^{s/2}\vert\hat{f}(\xi)\vert^2d\xi.
\]


\appendix
\section{The Weyl quantization of radial functions}
If $a(x,p)$ is a symbol in $\bbR^{2n}$, its Weyl quantization is the operator
$a^W(x,\h D)$ with kernel
\[
\calK_a(x,y) = \frac{1}{(2\pi\h)^n} \int e^{i\hinv (x-y)\cdot p} a\left(\frac{x+y}{2}, p\right)\, dp.
\]
The corresponding bilinear form $Q_a(f,g) = \inner{a^W(x,\h D)(f)}{\overline g}$ is
\begin{equation}\label{bil}
\forall f,\,g\in\calS(\bbR^n)\qquad Q_a(f,g) = 
\frac{1}{(2\pi\h)^n} \iiint e^{i\hinv (x-y)\cdot p} a\left(\frac{x+y}{2}, p\right)\,f(y)\, g(x) dp\,dx\,dy.
\end{equation}
It is not hard to see that 
\begin{equation}\label{soliaSerLema1}
Q_a(f,g) = \iint a(u,p)\, \calG(f,g)(u,p)\, du\, dp,
\end{equation}
where
	\begin{equation}\label{soliaSerLema2}
	\calG(f,g)(u,p) = \frac{1}{(\pi\h)^n}\int e^{2i\hinv v\cdot p} f(u-v)\,g(u+v)\, dv.
	\end{equation}

%
%
%
%
%
%
%
%
%
%
%

Here we gather some results on the Weyl quantization of radial functions on $T^*\bbR =\bbR^2$.  
Let $a\in \calS(\bbR^2)$ be a radial function, that is
\[
a(x,p) = \rho(r), \quad r=\sqrt{x^2+p^2}.
\]
To simplify notation let $A:=a^W(x,\h D)$.  By the equivariance of Weyl quantization with 
respect to the action of the symplectic (metaplectic) group, $A$ commutes with the 
quantum harmonic
oscillator $\calZ = -\h^2d^2/dx^2 + x^2$ and, 
by simplicity of the eigenvalues of the latter, the eigenfunctions $e_n$ of $\calZ$ 
are also eigenfunctions of $A$. 
Our goal is to compute the corresponding eigenvalues.   We follow the
argument in \cite{DHS}.

\medskip
One can show (starting with \S 13.1 of \cite{AW}, for example) that if one defines the functions $g_n(x)$ by the generating function
\begin{equation}\label{}
\pi^{-1/4}\,  e^{-x^2/2 +xt-t^2/4} =\sum_{n=0}^\infty \frac{t^n}{\sqrt{2^n n!}}\,g_n(x)
\end{equation}
then the $g_n$ are orthonormal in $L^2(\bbR)$ and satisfy
\[
-g_n''(x) + x^2 g_n(x) = (2n+1)\,g_n(x).
\]
For our problem we need the eigenfunctions of $\calZ$, so we need to re-scale
the variable $x$.  Define
\begin{equation}\label{}
e_n(x) := \h^{-1/4}\,g_n(x/\sqrt{\h}). 
\end{equation}
Then for each $\h$, $e_n$ is $L^2$-normalized and 
\begin{equation}\label{}
\h^{1/4}\,\left[-\h^2 \frac{d^2 \ }{dx^2} e_n(x) + x^2 e_n(x)\right] =
-\h\,g_n''(x/\sqrt{\h}) + \h\left(\frac{x}{\h}\right)^2 g_n(x/\sqrt{\h}) =
\h(2n+1) g_n(x/\sqrt{\h}).
\end{equation}
In other words, {\em the normalized eigenfunctions $e_n$ are given by the generating function}
\begin{equation}\label{}
G_t(x):=(\pi\h)^{-1/4}\,  e^{-x^2/2\h +xt/\sqrt{\h}-t^2/4} =\sum_{n=0}^\infty \frac{t^n}{\sqrt{2^n n!}}\,e_n(x,\h),
\end{equation}
where the notation emphasizes that $e_n$ also depends on $\h$.

We now use this generating function to compute the eigenvalues of $A$.
Note that
\[
Q_a(G_t, G_t) = \inner{A(G_t)}{G_t} = \sum_{n=0}^\infty \frac{t^{2n}}{2^n n!}\, \lambda_n
\]
where $\lambda_n = \inner{A(e_n)}{e_n}$
is the eigenvalue of $A$ corresponding to $e_n$.  Computing using (\ref{soliaSerLema1})
and (\ref{soliaSerLema2}):
\[
\calG(G_t, G_t)(x,p) = \frac{e^{-t^2/2}}{(\pi\h)^{3/2}} \int e^{2\hinv ivp}\, e^{-x^2/\h}\,e^{-v^2/\h}\,
e^{2xt/\sqrt{\h}}\, dv =
\]
\[
= \frac{e^{-t^2/2-x^2/\h + 2xt/\sqrt{\h}}}{(\pi\h)^{3/2}} \int e^{2\hinv ivp}\,e^{-v^2/\h}\,dv =
\]
\[
= \frac{1}{\pi\h}\,e^{-t^2/2-(x^2+p^2)/\h + 2xt/\sqrt{\h}},
\]
and therefore
\begin{equation}\label{laConc}
\sum_{n=0}^\infty \frac{t^{2n}}{2^n n!}\, \lambda_n = \frac{e^{-t^2/2}}{\pi\h}\, \int_{\bbR^2} a(x,p)
\,e^{-(x^2+p^2)/\h + 2xt/\sqrt{\h}}\, dx\,dp.
\end{equation}
Next we use that $a$ is radial and integrate in polar coordinates.  The key integral is
\begin{equation}\label{}
\int_0^{2\pi} e^{2tr\cos(\theta)/\sqrt{\h}}\, d\theta = 2\pi I_0(2tr/\sqrt{\h}),
\end{equation}
where $I_0$ is the modified Bessel function of order zero.  At this point
we can conclude that
\begin{equation}\label{yaCasi}
\sum_{n=0}^\infty \frac{t^{2n}}{2^n n!}\, \lambda_n = 
\frac{2e^{-t^2/2}}{\h}\, \int_0^\infty \rho(r) e^{-r^2/\h}\,I_0(2tr/\sqrt{\h})\, rdr.
\end{equation}

Now it is known that, for any $u\in\bbR$,
\begin{equation}\label{laConclusion}
I_0(s) = e^u\sum_{k=0}^\infty \frac{(-u)^k}{k!}\, L_k(s^2/4u)
\end{equation}
where the $L_k$ are the Laguerre polynomials (in particular the right-hand side is independent of $u$).
If we
take $u=t^2/2$, (\ref{laConclusion}) gives us that
\[
I_0(s) = e^{t^2/2}\sum_{k=0}^\infty \frac{(-t^2/2)^k}{k!} L_k(s^2/2t^2).
\]
Substituting back into (\ref{yaCasi}) we obtain
\begin{equation}\label{ufL}
\sum_{n=0}^\infty \frac{t^{2n}}{2^n n!}\, \lambda_n = \frac{2}{\h} \sum_{k\geq 0}
\int_0^\infty \rho(r) e^{-r^2/\h}\,\frac{(-t^2)^k}{2^k k!} L_k(2r^2/\h)\, rdr.
\end{equation}
Equating coefficients of like powers of $t$ we conclude that $\forall n$
\[
\lambda_n = \frac{(-1)^k2}{\h}\int_0^\infty \rho(r) e^{-r^2/\h}\,L_n(2r^2/\h)\, rdr.
\]
If we now let $u=r^2$, we finally get

\begin{equation}\label{ahoraSi}
\boxed{
	\lambda_n = \frac{(-1)^n}{\h} \int_0^\infty \rho(\sqrt{u}) e^{-u/\h}\,L_n(2u/\h)\, du .
}
\end{equation}


\bigskip
Although we do not need it for the proof of our main theorem, we note the following:
\begin{theorem}\label{evalAsymptotics}
	Let (as in the main body of the paper)
	\[
	\hbar = \frac{\mathcal E}{2n+1},\quad \mathcal{E} \text{ fixed},\quad  n = 1,2, \ldots ,
	\]
	Then, maintaining the previous notation, as $n\to\infty$
	\[
	\lambda_n = \rho(\sqrt{\calE}) + O(\h).
	\] 	
\end{theorem}
\begin{proof}
By the functional calculus the operator $\rho(\calZ^{1/2})$ is an $\h$ pseudo-differential
operator with principal symbol $\rho(r)$, that is, with the same principal 
symbol as $a^W$.  Therefore
\[
\lambda_n = \inner{a^W(e_n)}{e_n} = \inner{\rho(\calZ^{1/2})(e_n)}{e_n} + O(\h)
= \rho(\sqrt{\calE}) + O(\h).
\]
\end{proof}
In view of (\ref{ahoraSi}), we immediately obtain: 
\begin{corollary}
Let
\[
\psi_n(u) := \frac{(-1)^n}{\hbar}e^{-u/\hbar} L_n \left( \frac{2u}{\hbar} \right)
\]
so that $\lambda_n = \int_0^\infty \rho(\sqrt{u}) \psi_n (u)\, du$. 
Then, if $\h$ and $n$ are related as above, the sequence $(\psi_n)$ tends weakly to 
the delta function at $\calE$.
\end{corollary}

It is instructive to consider directly the behavior of the functions $\psi_n$.
As we will see, there is an oscillatory and a decaying region of $\psi_n$ (similar to the
Airy function). 
For a fixed $n$, $\psi_n$ has $n$ zeros. As $n$ increases, where do the zeros concentrate? According to \cite{gawronski1987asymptotic}, the zeros of $L_n$ are real and simple. 
Let us denote by $\lambda_{n,k}$ the zeros of $L_n$. According to \cite{gatteschi2002asymptotics}  (restricting to the case $\alpha =0$), the zeros $\lambda_{n,k}$ are in the oscillatory region 
\[
0 < x < \nu := 4n+2
\] and satisfy the following inequalities and asymptotic approximation:

\begin{theorem} (\cite{gatteschi2002asymptotics}) The first zero $\lambda_{n,1}$ satisfies 
\label{th:Ineq}
\[
0 < \lambda_{n,1} \le \frac{3}{2n+1}, n=1,2,\ldots.
\]
\end{theorem}

\begin{theorem} (\cite{gatteschi2002asymptotics}) For a fixed $m$, the zeros of $L_n$ satisfy
\[
\lambda_{n,n-m+1} = \nu +2^{1/3} a_m \nu^{1/3}+\frac{1}{5} 2^{4/3} a_m^{2}\nu^{-1/3}+O(n^{-1}),\text{ as }n \to \infty,
\]
where $a_m$ is the $m$-th negative zero of the Airy function, in decreasing order.
\end{theorem}

Let us now denote by $\mu_{n,k}$ the zeros of $\psi_n(u)$, so that $\mu_{n,k}= \frac{\hbar}{2}\lambda_{n,k}$. Substituting $\hbar = \mathcal E/(2n+1)$, Theorem \ref{th:Ineq} implies that the first zero satisfies
\[
\mu_{n,1} \le \frac{3}{2 \mathcal E} \hbar^2.
\]
On the other hand, the last zero satisfies
\[
\mu_{n,n} = \mathcal E+(\mathcal E/2)^{1/3}a_1 \hbar^{2/3} + \frac{(\mathcal E)^{-1/3} a_1^2}{5}\hbar^{4/3} + O(\hbar^2), \text{ as } \hbar \to 0.
\]

This implies that the first zero is close to 0 while the last one is close to $\mathcal E$ as $\hbar \to 0$. In fact, if we define 
\[
N_n(x) = \big | \left\{ k \in \{ 1,2,\ldots, n\} | \lambda_{n,k} \le x \right\} \big | , x\in \mathbb R,
\]
it can be shown (\cite{gawronski1987asymptotic}) that 
\[
\lim_{n\to \infty} \frac{1}{n} N_n(4 nx) =\frac{2}{\pi} \int_0^{x}t^{-1/2} (1-t)^{1/2}dt, \; \;\text{ for } 0 \le x \le 1.
\]

\begin{figure}[h!]
\begin{center}
{\includegraphics[width=0.49\textwidth]{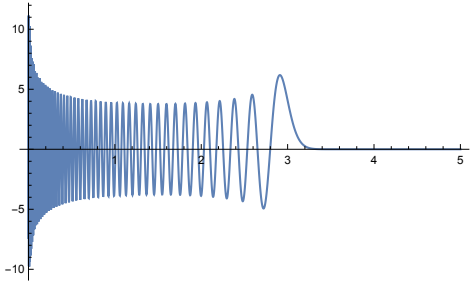}}
\end{center}
\caption{\label{fig:Gn100E3} Graph of $\psi_n$ in the interval $[0,5]$. Here $n=100$, $\mathcal E = 3$.}
\end{figure} 

We note that $\lambda_{n,k} \le 4 n x$ if and only if $\mu_{n,k} \le \mathcal E x \left( 1-\frac{1}{2n+1} \right)$. This implies that
\[
\lim_{n\to \infty} \frac{1}{n} \left | \left\{ k : \mu_{n,k} \le z \left( 1-\frac{1}{2n+1} \right)  \right\} \right |= \frac{2}{\pi} \int_0^{z/\mathcal E}t^{-1/2} (1-t)^{1/2}dt, \; \; 0 \le z \le \mathcal E.
\]
We note that the integral on the right-hand side is equal to one for $z=\mathcal E$. In particular, this shows that the zeros of $\psi_n$ ``cover'' the entire oscillatory region $[0,\mathcal E]$, asymptotically for $n$ large.  

Choosing $n=100$ and $\mathcal E = 3$, the corresponding graph of $\psi_n$ in the interval $[0,5]$ is shown in Figure \ref{fig:Gn100E3}. We can corroborate numerically that the zeros of $\psi_n$ are located in the oscillatory region $[0,\mathcal E]$. We can easily see that $L_n$ is always locally decreasing near the origin and locally increasing/decreasing around the last zero for $n$ even/odd. As a result, the last critical point of $\psi_n$ is always a local maximum.

\section{The Remainder in Taylor's theorem}\label{Taylor}

For completeness we include here the elementary derivation of the
expression for the remainder in Taylor's theorem that we used in the proof 
of Theorem \ref{simboloReducido}.
Let us start with a smooth one-variable function $f$, and write
\[
f(t) = f(\calE) + \int_0^1 \frac{d\ }{du}f(ut + (1-u)\calE)\, du = f(\calE)+
(t-\calE)\int_0^1 \,f'(ut + (1-u)\calE)\, du .
\]
So if we let
\begin{equation}\label{truc}
g(t):= \int_0^1 f'(ut + (1-u)\calE)\, du,
\end{equation}
then $g$ is smooth and $f(t) = f(\calE) + (t-\calE)g(t)$.
Repeating the argument with $f$ replaced by $g$ we obtain that
\[
g(t) = g(\calE) + (t-\calE)R(t)
\]
where
\[
R(t) = \int_0^1 g'(vt + (1-v)\calE)\, dv.
\]
Since $g(\calE) = f'(\calE)$, substituting we obtain $f(t) = f(\calE) + (t-\calE)f'(\calE) + (t-\calE)^2 R(t)$,
as desired.
Finally we compute the remainder $R(t$).
Using (\ref{truc}), 
\[
g'(x) = \int_0^1 uf''(ux + (1-u)\calE)\, du,
\]
and therefore
\[
R(t) = \int_0^1 \int_0^1 uf''[u(vt + (1-v)\calE)+ (1-u)\calE] du\, dv.
\]




\end{document}